\documentclass[onecolumn]{IEEEtran}
\usepackage[utf8]{inputenc} 
\usepackage[T1]{fontenc} 
\usepackage{amsthm}
\usepackage[cmex10]{amsmath}
\usepackage{times,amssymb,amsfonts,float,nicefrac,color,mathrsfs,caption} 
\usepackage{enumerate,multirow,caption,tikz,graphicx,enumitem,soul}
\usepackage[mathscr]{eucal}
\usepackage{sidecap}
\usepackage{algpseudocode}
\usepackage{verbatim}
\usepackage{epstopdf}
\usepackage{textcomp}
\usetikzlibrary{shapes,arrows}
\usepackage[noadjust]{cite}
\usepackage{booktabs}
\allowdisplaybreaks

\usepackage[font={it}]{caption}
\usepackage[margin=0.05in]{subcaption}
\usepackage[bottom]{footmisc}
\usepackage[normalem]{ulem}

\usepackage[ruled,vlined,linesnumbered]{algorithm2e}

\usepackage{hyperref} 

\usepackage{balance}

\newtheoremstyle{mytheoremstyle}{3pt}{3pt}{\itshape}{}{\bf}{.}{.3em}{} 
\theoremstyle{mytheoremstyle}
\newtheorem{theorem}{Theorem}

\newcommand\nc\newcommand
\nc\bfa{{\boldsymbol a}}\nc\bfA{{\boldsymbol A}}\nc\cA{{\mathscr A}}
\nc\bfb{{\boldsymbol b}}\nc\bfB{{\boldsymbol B}}\nc\cB{{\mathscr B}}
\nc\bfc{{\boldsymbol c}}\nc\bfC{{\boldsymbol C}}\nc\cC{{\mathscr C}}
\nc\bfd{{\boldsymbol d}}\nc\bfD{{\boldsymbol D}}\nc\cD{{\mathscr D}}
\nc\bfe{{\boldsymbol e}}\nc\bfE{{\boldsymbol E}}\nc\cE{{\mathscr E}}
\nc\bff{{\boldsymbol f}}\nc\bfF{{\boldsymbol F}}\nc\cF{{\mathscr F}}
\nc\bfg{{\boldsymbol g}}\nc\bfG{{\boldsymbol G}}\nc\cG{{\mathscr G}}
\nc\bfh{{\boldsymbol h}}\nc\bfH{{\boldsymbol H}}\nc\cH{{\mathscr H}}
\nc\bfi{{\boldsymbol i}}\nc\bfI{{\boldsymbol I}}\nc\cI{{\mathcal I}}
\nc\bfj{{\boldsymbol j}}\nc\bfJ{{\boldsymbol J}}\nc\cJ{{\mathscr J}}
\nc\bfk{{\boldsymbol k}}\nc\bfK{{\boldsymbol K}}\nc\cK{{\mathscr K}}
\nc\bfl{{\boldsymbol l}}\nc\bfL{{\boldsymbol L}}\nc\cL{{\mathscr L}}
\nc\bfm{{\boldsymbol m}}\nc\bfM{{\boldsymbol M}}\nc{\cM}{{\mathscr M}}
\nc\bfn{{\boldsymbol n}}\nc\bfN{{\boldsymbol N}}\nc\cN{{\mathscr N}}
\nc\bfo{{\boldsymbol o}}\nc\bfO{{\boldsymbol O}}\nc\cO{{\mathscr O}}
\nc\bfp{{\boldsymbol p}}\nc\bfP{{\boldsymbol P}}\nc\cP{{\mathscr P}}\nc\eP{{\EuScriptP}}\nc\fP{{\mathfrak P}}
\nc\bfq{{\boldsymbol q}}\nc\bfQ{{\boldsymbol Q}}\nc\cQ{{\mathscr Q}}
\nc\bfr{{\boldsymbol r}}\nc\bfR{{\boldsymbol R}}\nc\cR{{\mathscr R}}
\nc\bfs{{\boldsymbol s}}\nc\bfS{{\boldsymbol S}}\nc\cS{{\mathscr S}}
\nc\bft{{\boldsymbol t}}\nc\bfT{{\boldsymbol T}}\nc\cT{{\mathscr T}}
\nc\bfu{{\boldsymbol u}}\nc\bfU{{\boldsymbol U}}\nc\cU{{\mathscr U}}
\nc\bfv{{\boldsymbol v}}\nc\bfV{{\boldsymbol V}}\nc\cV{{\mathscr V}}
\nc\bfw{{\boldsymbol w}}\nc\bfW{{\boldsymbol W}}\nc\cW{{\mathscr W}}
\nc\bfx{{\boldsymbol x}}\nc\bfX{{\boldsymbol X}}\nc\cX{{\mathscr X}}
\nc\bfy{{\boldsymbol y}}\nc\bfY{{\boldsymbol Y}}\nc\cY{{\mathscr Y}}
\nc\bfz{{\boldsymbol z}}\nc\bfZ{{\boldsymbol Z}}\nc\cZ{{\mathscr Z}}

\newtheorem{lemma}[theorem]{Lemma}

\newtheorem{proposition}[theorem]{Proposition}

\newtheorem{corollary}[theorem]{Corollary}
\newtheorem{definition}{Definition}

\newtheorem{construction}{Construction}[section]

\theoremstyle{remark}
\newtheorem{remark}{Remark}[section]

\DeclareMathOperator{\rank}{rank}

\DeclareMathOperator{\Span}{Span}
\newcommand{\ff}{{\mathbb F}}

\begin{document}
	
	\title{A lower bound on the field size of convolutional codes with a maximum distance profile and an improved construction}
	
	\author{\IEEEauthorblockN{Zitan Chen}}
	\maketitle	
	
	{\renewcommand{\thefootnote}{}\footnotetext{
			
			\vspace{-.2in}
			
			\noindent\rule{1.5in}{.4pt}

			{	
				The author is with the School of Science and Engineering, Future Networks of Intelligence Institute, The Chinese University of Hong Kong, Shenzhen, China. Email: chenztan@cuhk.edu.cn
				
				This work was supported in part by 
				the Basic Research Project of Hetao Shenzhen-Hong Kong Science and Technology Cooperation Zone under Project HZQB-KCZYZ-2021067, 
				the Guangdong Provincial Key Laboratory of Future Network of Intelligence under Project 2022B1212010001, 
				the National Natural Science Foundation of China under grant 62201487,
				and the Shenzhen Science and Technology Program (Grant No. RCBS20221008093108032).
			}
		}
	}
	\renewcommand{\thefootnote}{\arabic{footnote}}
	\setcounter{footnote}{0}

	\begin{abstract} Convolutional codes with a maximum distance profile attain the largest possible column distances for the maximum number of time instants and thus have outstanding error-correcting capability especially for streaming applications. Explicit constructions of such codes are scarce in the literature. In particular, known constructions of convolutional codes with rate \(k/n\) and a maximum distance profile require a field of size at least exponential in \(n\) for general code parameters. At the same time, the only known lower bound on the field size is the trivial bound that is linear in \(n\). In this paper, we show that a finite field of size \(\Omega_L(n^{L-1})\) is necessary for constructing convolutional codes with rate \(k/n\) and a maximum distance profile of length \(L\). As a direct consequence, this rules out the possibility of constructing convolutional codes with a maximum distance profile of length \(L\geq 3\) over a finite field of size \(O(n)\).
		
	Additionally, we also present an explicit construction of convolutional code with rate $k/n$ and a maximum profile of length $L=1$ over a finite field of size $O(n^{\min\{k,n-k\}})$, achieving a smaller field size than known constructions with the same profile length.
	\end{abstract}
	
	
	\section{Introduction}
	Convolutional coding is a time-dependent coding method that encodes the present information word with a small number of previous information words to generate the present codeword. An important family of convolutional codes are those with a maximum distance profile (MDP), also known as MDP convolutional codes, which have excellent error-correcting capability when used in streaming applications. To introduce MDP convolutional codes, let us begin with presenting a few basic concepts and results in convolutional coding from a module-theoretic perspective. 
	
	Let $\ff$ be a finite field and $\ff[D]$ be the ring of polynomials with indeterminate $D$ and coefficients in $\ff$. Let $n>k>0$ be integers. 
	
	\begin{definition}
		An $(n,k)$ convolutional code over $\ff$ is an $\ff[D]$-submodule $\cC\subset \ff[D]^n$ of rank $k$. A $k\times n$ polynomial matrix $G(D) = (g_{ij}(D))$ over $\ff[D]$ is called a generator matrix of the code $\cC$ if 
		\begin{align*}
			\cC= \{u(D)G(D)\mid u(D)\in \ff[D]^k\}.
		\end{align*}
		An $(n-k)\times n$ polynomial matrix $H(D)=(h_{ij}(D))$ over $\ff[D]$ is called a parity check matrix of the code $\cC$ if 
		\begin{align*}
			\cC = \{v(D)\in\ff[D]^n\mid H(D)v(D)=0\}.
		\end{align*}
	\end{definition} 
	
	A $k\times n$ polynomial matrix $G(D)$ is called \emph{basic} if it has a polynomial right inverse. The \emph{constraint length for the $i$th input} or the \emph{$i$th row degree} of the matrix $G(D)$ is defined to be 
	\begin{align*}
		\nu_i=\max_{1\leq j\leq n}\{\deg g_{ij}(D)\}, \quad i=1,\ldots,k.
	\end{align*}
	Furthermore, the \emph{memory} $m$ of the matrix $G(D)$ is defined as 
	\begin{align*}
		m=\max_{1\leq i\leq k}\{\nu_i\}.
	\end{align*}
	The \emph{overall constraint length} of the matrix $G(D)$ is defined to be 
	\begin{align*}
		\nu=\sum_{i=1}^{k}\nu_i.
	\end{align*}
	
	A generator matrix $G(D)$ of an $(n,k)$ convolutional code $\cC$ is called \emph{minimal} or \emph{reduced} if its overall constraint length is minimal over all generator matrices of the code $\cC$.
	This minimum overall constraint length is called the \emph{degree} $\delta$ of the code $\cC$, i.e.,
	\begin{align*}
		\delta = \min\{\nu \mid G(D)\text{ is a generator matrix of }\cC\}.
	\end{align*}
	We shall call an $(n,k)$ convolutional code with degree $\delta$ an $(n,k,\delta)$ convolutional code. 
	
	The degree $\delta$ is arguably one of the most fundamental parameters of a convolutional code as it stipulates the smallest number of memory elements needed to realize the code and is closely related to the decoding complexity \cite{johannesson2015fundamentals}. 
	One may resort to the following lemma, which provides a necessary and sufficient condition for a generator matrix to be minimal, to determine the degree of a convolutional code.
	\begin{lemma}[\cite{mceliece1993general}]\label{le:minimal-matrices}
		Let $G(D)$ be a $k\times n$ matrix over $\ff[D]$ and define the matrix of the highest order coefficients for $G(D)$, denoted by $\bar{G}=(\bar{G}_{ij})$, by
		\begin{align*}
			\bar{G}_{ij}=\mathrm{coeff}_{D^{\nu_i}}g_{ij}(D),
		\end{align*} where $\mathrm{coeff}_{D^{\nu_i}}g_{ij}(D)$ denotes the coefficient of $D^{\nu_i}$ in the polynomial $g_{ij}(D)$. Then $G(D)$ is minimal if and only if $\bar{G}$ has rank $k$.
	\end{lemma}
	
	Two basic distance measures can be used to assess the error-correcting capability of convolutional codes: the \emph{free distance} and the \emph{column distances}.
	Note that for $v(D)\in \ff[D]^n$, we may write $v(D)=\sum_{j\in\mathbb{N}}v_j D^j\in \ff^n[D]$. Define the \emph{Hamming weight} of $v(D)$ to be 
	\begin{align*}
		\mathrm{wt}(v(D)) = \sum_{j\in\mathbb{N}}\mathrm{wt}(v_j),
	\end{align*} where $\mathrm{wt}(v_j)$ denotes the Hamming weight of the length-$n$ vector $v_j$ over $\ff$. The {free distance} of a convolutional code $\cC$ is then defined to be 
	\begin{align*}
		d(\cC)=\min\{\mathrm{wt}(v(D))\mid v(D)\in\cC, v(D)\neq 0\}.
	\end{align*}
	It is shown in \cite{rosenthal1999maximum} that the free distance of an $(n,k,\delta)$ convolutional code satisfies the following upper bound that generalizes the Singleton bound for block codes:
	\begin{align}
		d(\cC)\leq (n-k)\left( \left\lfloor\frac{\delta}{k}\right\rfloor+1\right)+\delta+1.\label{eq:sb}
	\end{align}
	Convolutional codes whose free distance attains \eqref{eq:sb} with equality are called maximum distance separable (MDS) convolutional codes.
	
	To define the column distances, let us first introduce the notion of \emph{truncated generator matrices}. To this end, let $G(D)=\sum_{i=0}^{m}G_iD^i$ be a generator matrix of memory $m$ for an $(n,k)$ convolutional code $\cC$. 
	Let $j\geq 0$ be an integer. Then the \emph{$j$th truncated generator matrix} is defined to be
	\begin{align*}
		G_j^c=\begin{pmatrix}
			G_0 & G_1 & \cdots & G_j\\
			& G_0 & \cdots & G_{j-1}\\
			&	& \ddots & \vdots\\
			&	& 		 & G_0
		\end{pmatrix},
	\end{align*} where $G_i=0$ for $i>m$ and the empty blocks denote zero matrices. Similarly, let $H(D)$ be a parity check matrix of memory $m'$ for the code $\cC$. The \emph{$j$th truncated parity check matrix} is given by
	\begin{align*}
		H_j^c=\begin{pmatrix}
			H_0	&	   & 		 & \\
			H_1	& H_0  & 		 & \\
			\vdots	& \vdots & \ddots & \\
			H_j & H_{j-1} & \cdots & H_0
		\end{pmatrix}, 
	\end{align*} where $H_i=0$ for $i>m'$.
	With the $j$th truncated generator matrix, the $j$th column distance of a convolutional code can be defined as follows.
	\begin{definition}
		Let $G(D)=\sum_{i=0}^{m}G_iD^i$ be a generator matrix of memory $m$ for an $(n,k)$ convolutional code $\cC$ such that $G_0$ has full rank. For $j\geq 0$, the $j$th column distance of the code $\cC$ is given by 
		\begin{align*}
			d_j^c(\cC)=\min_{\substack{u_i\in\ff^k,i=0,\ldots,j,\\u_0\neq 0}}\{\mathrm{wt}\big((u_0,\ldots,u_j) G_j^c\big)\}.
		\end{align*}
	\end{definition}
	\begin{theorem}[\cite{gluesing2006strongly}]\label{thm:sb-c}
		Let $\cC$ be an $(n,k)$ convolutional code. For all $j\geq 0$ the column distance satisfies 
		\begin{align}
			d_j^c(\cC)\leq (n-k)(j+1)+1.\label{eq:sb-c}
		\end{align} 
		Moreover, equality for a given $j$ implies that all the previous distances $d_i^c(\cC),i\leq j$ also attain their versions of the bound \eqref{eq:sb-c} with equality.
	\end{theorem}
	
	Convolutional codes whose column distances attain the bound \eqref{eq:sb-c} for the largest possible number of time instants are called MDP convolutional codes. Another closely related family of convolutional codes are those codes whose column distances reach the bound \eqref{eq:sb} at the earliest possible time instant, termed strongly-MDS convolutional codes. These two types of distance-optimal convolutional codes are formally defined as follows.
	
	\begin{definition}[\cite{hutchinson2005convolutional}, \cite{gluesing2006strongly}]\label{def:sMDS-MDP}
		Let $\cC$ be an $(n,k,\delta)$ convolutional code. Let $M=\lfloor\frac{\delta}{k}\rfloor+\lceil\frac{\delta}{n-k}\rceil$ and $L=\lfloor\frac{\delta}{k}\rfloor+\lfloor\frac{\delta}{n-k}\rfloor$.
		\begin{enumerate}
			\item The code $\cC$ is said to be strongly-MDS if
			\begin{align*}
				d_M^c(\cC)=(n-k)\left(\left\lfloor\frac{\delta}{k}\right\rfloor+1\right)+\delta+1.
			\end{align*}
			\item The code $\cC$ is said to be MDP if 
			\begin{align*}
				d_L^c(\cC)=(n-k)(L+1)+1.
			\end{align*}
		\end{enumerate}
	\end{definition}
	Clearly, strongly-MDS convolutional codes are also MDS by definition. Although the column distance of strongly-MDS convolutional codes attains the bound \eqref{eq:sb} at the earliest possible time instant $M$, their column distance at the time instant prior to $M$ may not attain the bound \eqref{eq:sb-c}. On the contrary, the column distances of MDP convolutional codes attain the bound \eqref{eq:sb-c} for the maximum number of time instants, i.e., $L+1$ time instants, but their free distance may not attain the bound \eqref{eq:sb}. As we will refer to the parameter $L$ frequently in the sequel, let us call $L$ the \emph{maximum profile length}.
	
	Before proceeding to discuss prior work on MDP convolutional codes and the contributions of this work, let us mention a characterization of 
	optimal column distances by the determinants of full-size square submatrices of $G_j^c$ and $H_j^c$.
	
	\begin{theorem}[\cite{gluesing2006strongly}]\label{thm:cd}
		Let $G(D)$ be a $k\times n$ basic and minimal generator matrix of an $(n,k)$ convolutional code $\cC$ and let $H(D)$ be a $(n-k)\times n$ basic parity check matrix of the code. 
		Then the following are equivalent:
		\begin{enumerate}
			\item $d_j^c(\cC)=(n-k)(j+1)+1$;
			\item every $k(j+1)\times k(j+1)$ full-size minor of $G_j^c$ formed by the columns with indices $1\leq t_1 < \ldots < t_{k(j+1)}$ where $t_{ks+1}\geq ns+1$ for $s=1,\ldots,j$ is nonzero;
			\item every $(n-k)(j+1)\times (n-k)(j+1)$ full-size minor of $H_j^c$ formed by the columns with indices $1\leq t_1<\ldots <t_{(n-k)(j+1)}$ where $t_{ks}\leq ns$ for $s=1,\ldots,j$ is nonzero.
		\end{enumerate}
	\end{theorem}
	
	Theorem~\ref{thm:cd} suggests that one way of constructing MDP convolutional codes is to design a minimal and basic generator matrix such that the corresponding truncated matrix has the property described in Item~2 of Theorem~\ref{thm:cd}. For convenience, we shall call Item~2 of Theorem~\ref{thm:cd} the \emph{MDP property} of $G_j^c$, and Item~3 of Theorem~\ref{thm:cd} the \emph{MDP property} of $H_j^c$.
	Note that the MDP property of $G_j^c$ implies that for {a convolutional} code to have optimal $j$th column distance, the full-size square submatrices that should have nonzero determinants are only those $k(j+1)\times k(j+1)$ submatrices formed by \emph{at most} $ks$ columns of the \emph{first} $ns$ columns of $G_j^c$ for $s=1,\ldots,j$.\footnote{Equivalently, the full-size square submatrices that should have nonzero determinants are only those $k(j+1)\times k(j+1)$ submatrices formed by \emph{at least} $ks$ columns of the \emph{last} $ns$ columns of $G_j^c$ for $s=1,\ldots,j$.} In fact, the determinant of any other full-size square submatrix of $G_j^c$ is always zero.
	
	So the problem of constructing MDP convolutional codes boils down to constructing a minimal and basic generator matrix with the MDP property. By means of Lemma~\ref{le:minimal-matrices}, minimal generator matrices can be constructed without much effort. However, it is in general not obvious to construct basic generator matrices. Nevertheless, it is observed that if certain divisibility condition holds then the assumption of being basic in Theorem~\ref{thm:cd} can be replaced by a milder one.
	
	\begin{lemma}[\cite{alfarano2020left}]\label{le:mdp-generic}
		Let $\delta$ be an integer such that $k\mid \delta$
		and let $G(D)$ be a $k\times n$ minimal generator matrix with row degree $\nu_i=\delta/k$. If $G_L^c$ has the MDP property then $G(D)$ is basic and it generates an $(n,k,\delta)$ MDP convolutional code.
	\end{lemma}
	
	In view of Lemma~\ref{le:mdp-generic}, the task of constructing MDP convolutional codes in the case of $k\mid \delta$ simplifies to designing minimal generator matrices with the MDP property. Lastly, we note that, similar to the MDS property of block codes, a convolutional code has a maximum distance profile if and only if its dual has a maximum distance profile. This implies that to construct MDP convolutional codes with a given degree for all rates, it suffices to consider the codes with rate at most half.
	\begin{theorem}[\cite{gluesing2006strongly}]\label{thm:mdp-dual}
		An $(n,k,\delta)$ convolutional code $\cC$ over $\ff$ with generator matrix $G(D)$ is MDP if and only if its dual code $\cC^\perp$ with parity check matrix $G(D)$ is an $(n,n-k,\delta)$ MDP convolutional code over $\ff$.
	\end{theorem}
	
	{Finally, we note that asymptotic notation is used throughout the text. Given two functions $f(n),g(n)$, we use $f(n)=O(g(n))$ and $g(n)=\Omega(f(n))$ to denote $\limsup_{n\to\infty}(f(n)/g(n))<\infty$, and write $f(n)=O_t(g(n))$ and $g(n)=\Omega_t(f(n))$ to emphasize the hidden multiplicative constants in the asymptotic notation depend on the parameter $t$.}
	
	\subsection{Related work}
	
	There are only a few explicit constructions of MDP convolutional codes in the literature. In \cite{gluesing2006strongly}, the authors presented the first explicit construction of MDP convolutional codes over a large finite field of large characteristic. The idea that is instrumental in this construction is to first construct lower triangular Toeplitz {superregular} matrices\footnote{Roughly speaking, a lower triangular Toeplitz matrix is called {superregular} if its minors that are not trivially zero are nonzero. 
	} and then utilize them as the building blocks for the construction of MDP convolutional codes. The same idea was also explored in \cite{almeida2013new} to construct another class of MDP convolutional codes over a large finite field of arbitrary characteristic. A third general class of explicit MDP convolutional codes was presented in \cite{alfarano2020weighted}, which uses generalized Reed-Solomon (RS) codes as the building blocks. Two families of MDP convolutional codes with $L=1$ were recently constructed in \cite{chen2022convolutional} and \cite{luo2023construction}, relying on techniques from the theory of skew polynomials \cite{lam1985general}, \cite{lam1988vandermonde} and linearized polynomials, respectively.
	The aforementioned families of codes are the only known algebraic constructions of MDP convolutional codes. The size of the finite field that suffices for constructing these codes with rate $R=k/n$ and degree $\delta$ are summarized in Table~\ref{tab:comp}.
	
	As for impossibility results, the only known lower bound on the field size that is necessary for the existence of $(n,k,\delta)$ MDP convolutional codes is the trivial bound $\Omega(n)$. Therefore, it remained possible that $(n,k,\delta)$ MDP convolutional codes could be constructed over a finite field of size $O(n)$ for general code parameters.
	
	\begin{table}[!t]
		\renewcommand{\arraystretch}{1.3}
		\centering
		\begin{tabular}{|c|c|c|}
			\hline
			Field size  & Maximum profile length & Reference\\
			\hline\hline
			$2^{O((R^{-1}\delta+n)^2)}$ & $L\geq 1$ & \cite{gluesing2006strongly}\\
			\hline
			$2^{O(2^{R^{-1}\delta+n})}$ & $L\geq 1$ & \cite{almeida2013new}\\
			\hline
			$2^{O((\frac{\delta^3}{Rn}+\delta Rn)\log n)}$ & $L\geq 1$ & \cite{alfarano2020weighted}\\
			\hline
			$O(n^{2\delta})$ & $L=1$ & \cite{chen2022convolutional}\\
			\hline
			$O((3n/2)^\delta)$ & $L=1$ & \cite{luo2023construction}\\
			\hline
			$O(n^\delta)$ & $L=1$ & This paper\\
			\hline
		\end{tabular}
		\caption{A comparison of the field size requirement for constructions of MDP convolutional codes with rate $R=k/n$ and degree $\delta$.}
		\label{tab:comp}
	\end{table}

	\subsection{Our results}
	In this paper we derive the first nontrivial {\emph{asymptotic}} lower bound on the field size of MDP convolutional codes for general parameters. Specifically, we show that a finite field of size \[\Omega_L(n^{L-1})\] is necessary for constructing MDP convolutional codes with rate $k/n$ and a maximum distance profile of length $L$. 
	As a direct consequence, this rules out the possibility of constructing convolutional codes with a maximum distance profile of length $L\geq 3$ over a finite field of size $O(n)$ for general code parameters.
	
	Additionally, we present an explicit construction of $(n,k)$ MDP convolutional codes over a finite field of size $O(n^{\delta})$ with degree $\delta=\min\{k,n-k\}$, improving on the known explicit algebraic constructions of convolutional codes with maximum profile length $L=1$. This construction is a refinement of the construction in \cite{chen2022convolutional} that relies on skew polynomials. More precisely, we show that by carefully choosing the evaluation points for skew polynomials, one can reduce the field size of the code construction in \cite{chen2022convolutional} from $O(n^{2\delta})$ to $O(n^{\delta})$.
	
	\section{The bound}
	
	In this section we show that a finite field $\ff_q$ of size $q=\Omega_L(n^{L-1})$ is necessary for any $(n,k,\delta)$ MDP convolutional codes, where $L=\lfloor\frac{\delta}{k}\rfloor+\lfloor\frac{\delta}{n-k}\rfloor$ is the maximum profile length. 
	
	As mentioned before, it is clear that the existence of $(n,k,\delta)$ MDP convolutional codes over $\ff_q$ demands $q$ growing at least linearly in $n$. The following proposition formalizes this simple observation.
	
	\begin{proposition}\label{prop:bound}
		Let $\cC$ be an $(n,k,\delta)$ MDP convolutional code over $\ff_q$. Then $q=\Omega(n)$.
	\end{proposition}
	\begin{proof}
		Note that the maximum profile length of $\cC$ satisfies $L\geq 0$. Let $G(D)$ be a generator matrix of $\cC$. Since $\cC$ is MDP, then by Theorem~\ref{thm:sb-c} we have $d_0^c(\cC)=n-k+1$. Thus, any $k$ columns of the $k\times k$ truncated matrix $G_0^c=G_0$ have full rank. In other words, $G_0$ generates an $(n,k)$ MDS block code over $\ff_q$ and it follows that $q=\Omega(n)$.
	\end{proof}
	
	In the following we prove a new {asymptotic} lower bound on the field size by showing that the MDP property cannot be satisfied if the field size is not large enough. Our proof ideas are as follows. For an $(n,k,\delta)$ convolutional code to be MDP, by Theorem~\ref{thm:cd}, it has to satisfy that the full-size square submatrices of $G_L^c$ that are formed by at least $kj$ columns of the last $nj$ columns of $G_L^c$ for $j=1,\ldots,L$ are all of full rank. Inspired by the techniques in \cite{gopi2020maximally}, using the probabilistic method, we show that for constant $k$ and growing $n$, there exists a selection of $k(L+1)$ columns of $G_L^c$ that span a $(k(L+1)-1)$-dimensional subspace of $\ff_q^{k(L+1)}$ if the field size is not large enough. This leads to a lower bound for any MDP convolutional code with vanishing rate $R=k/n$. To extend this result to \emph{any} rate $R$, we note that the dual code of a convolutional code $\cC$ with rate approaching one that possesses a basic generator matrix has the same degree as the code $\cC$ and a vanishing rate. It follows that for MDP convolutional codes with rate approaching one, the field size should also be large enough. Moreover, given an MDP convolutional code of arbitrary rate, one may puncture the code at appropriate coordinates to obtain a new code with rate approaching one, thereby enabling the field size constraint to come into effect.  
	
	As the first step, let us establish a lower bound for MDP convolutional codes with vanishing rate.
	\begin{lemma}\label{le:vanishing-rate}
		Let $k$ be a constant independent of $n$.
		Let $G(D)$ be a $k\times n$ basic and minimal generator matrix of an $(n,k,\delta)$ MDP convolutional code over $\ff_q$. If the maximum profile length is a constant $L>1$, then the field size satisfies
		\begin{align*}
			q=\Omega_k(n^{k/\lceil k/L\rceil-1}).
		\end{align*} 
	\end{lemma}
	
	\begin{proof}
		Without loss of generality, we may assume $R\leq 1/2$, i.e., $2k\leq n$. It follows that $n\geq k+k/L$.
		Let $A_{1},A_{2},\ldots,A_{L}$ be subsets of $\{1,\ldots,n\}$ such that the following holds:
		\begin{enumerate}
			\item $\sum_{j=1}^{L}|A_j|=k(L+1)$;
			\item either $|A_j|=k+\lfloor \frac{k}{L} \rfloor$ or $|A_j|=k+\lceil \frac{k}{L} \rceil$ for $j=1,\ldots,L$.
		\end{enumerate}
		For simplicity, let us write $k_1=|A_1|-k,\ldots,k_L=|A_L|-k$.
		
		Let $P$ be the $k(L+1)\times k(L+1)$ matrix formed by columns of $G_L^c$ with indices in the set $\bigcup_{j=1}^{L}(A_{j}+jn)$ where $A_{j}+jn$ means the set obtained by adding every element of $A_{j}$ by the integer $jn$. 
		For clarity, let us write the matrix $P$ explicitly as
		\begin{align*}
			P=\begin{pmatrix}
				G_{1,A_{1}} & G_{2,A_{2}}   & \dots  & G_{L,A_L}   \\
				G_{0,A_{1}} & G_{1,A_{2}}   & \dots  & G_{L-1,A_L} \\
				& G_{0,A_{2}}   & \dots  & G_{L-2,A_L} \\
				&                 & \ddots & \vdots      \\
				&                 &        & G_{0,A_L}
			\end{pmatrix},
		\end{align*} where the $k\times |A_j|$ matrix $G_{i,A_j}$ is the submatrix of $G_{i}$ with column indices in $A_j$ for $i=0,\ldots,L$ and $j=1,\ldots,L$.
		Note that $P$ is formed by at least $kj$ columns of the last $nj$ columns of $G_L^c$ for $j=1,\ldots,L$. Therefore, by the MDP property of $G_L^c$, the matrix $P$ should have rank $(L+1)k$. 
		
		Let $B_j$ be a $k$-subset of $A_j$ for $j=1,\ldots,L$. By elementary column operations, we can transform $P$ to 
		\begin{align*}
			\left(
			\begin{array}{cc|cc|c|cc}
				S_{1} & G_{1,B_{1}} & S_{2}   & G_{2,B_{2}}   & \dots  & S_{L} & G_{L,B_L}   \\
				0           & G_{0,B_{1}} & 0             & G_{1,B_{2}}   & \dots  & 0           & G_{L-1,B_L} \\
				&             & 0             & G_{0,B_{2}}   & \dots  & 0           & G_{L-2,B_L} \\
				&             &               &                 & \ddots & \vdots      & \vdots      \\
				&             &               &                 &        & 0           & G_{0,B_L}
			\end{array}
			\right),
		\end{align*}
		where each column of the $k\times k_j$ matrix $S_{j}$ is a linear combination of the corresponding column in $G_{j,A_{j}\setminus B_{j}}$ and the columns of the $k\times k$ matrices $G_{1,B_{1}},\ldots,G_{j,B_{j}}$ for $j=1,\ldots,L$.
		Denote $S=\begin{pmatrix}
			S_1 & S_2 & \ldots & S_L
		\end{pmatrix}$.
		
		Since the code $\cC$ is MDP, by Theorem~\ref{thm:sb-c}, we have $d_j^c(\cC)=(n-k)(j+1)+1$ for $j=0,\ldots,L$. Thus, the matrix $G_0$ generates an $(n,k)$ MDS block code. It follows that $\rank G_{0,B_j}=k$ for $j=1,\ldots,L$. Moreover, since $\rank P=(L+1)k$, we have $\rank S=k$.
		
		Let $\cS_j= \Span_{\ff_q} S_j$. Then $\cS_j$ is an $k_j$-dimensional subspace of $\ff_q^k$ and we have
		$
		\bigoplus_{j=1}^L \cS_j = \ff_q^{k}.
		$ 
		Note that the subspace $\cS_j$ is determined by the choice of $A_1,\ldots, A_j$. {In particular, fixing $A_1,\ldots,A_j$, the subspace $\cS_j$ is independent of the choice of $B_1,\ldots,B_j$. We formalize this observation as the following proposition, whose proof can be found in Appendix~\ref{app:a}.
		\begin{proposition}\label{prop:s}
			Fix $A_1,\ldots,A_j$ and let $B_i,B_i'$ be $k$-subsets of $A_i$ for $i=1,\ldots,j$. Let $S_i$ (resp., $S_i'$) be obtained with respect to $B_1,\ldots,B_i$ (resp., $B_1',\ldots,B_i'$). Then we have $\Span_{\ff_q} S_i=\Span_{\ff_q} S_i'$ for $i=1,\ldots,j$. 
		\end{proposition}
		} 
		Below we show that if the field size $q$ is not large enough, there exists a choice of $A_1,\ldots, A_L$ such that the matrix $S$ is not of full rank. 
		Let {$x$ be a vector uniformly distributed over $\ff_q^{k}$}.
		and let $\cT$ be the orthogonal complement of $\Span_{\ff_q}\{x\}$. Note that $\cT$ is a random hyperplane in $\ff_q^{k}$ if $x\neq 0$. 
		
		Let $\cA_j$ be the set of $|A_j|$-subsets of $\{1,\ldots,n\}$ that contain $\{1,\ldots,2k_j\}$, namely, $\cA_j:=\{A\subset\{1,\ldots,n\}\mid A \supset \{1,\ldots,2k_j\}, |A|=|A_j|\}$.
		Let $X_j$ be the number of $A_j\in \cA_j$ such that $\cS_j$ is contained in $\cT$ given fixed sets $A_1,\ldots,A_{j-1}$. We would like to show that if $q$ is too small, then $X_j>0$ with high probability over the set $\cA_j$. Then by the union bound, we will be able to claim the event that $X_j>0$ for all $j$ simultaneously and $\cT$ is a hyperplane in $\ff_q^k$ occurs with high probability. 
		
		We will estimate $\Pr[X_j>0]$ by the second moment method. Since $\cS_j$ is an $k_j$-dimensional subspace of $\ff_q^{k}$, by linearity of expectation we have 
		\begin{align*}
			\mathbb{E}[X_j]&=\sum_{A_j\in \cA_j}\Pr[\cS_j \subset \cT]\\
			&=\binom{{n-2k_j}}{k-k_j}\frac{1}{q^{k_j}}.
		\end{align*}
		Let $A'_j\in\cA_j$ and let $\cS_j'$ be the subspace determined by $A_1,\ldots,A_{j-1}$ and $A'_j$. Then the second moment of $X_j$ is given by 
		\begin{align*}
			\mathbb{E}[X_j^2]&=\sum_{A_j,A'_j\in \cA_j}\Pr[\cS_j, \cS'_j \subset \cT]\\
			&=\sum_{i=2k_j}^{k+k_j}\sum_{\substack{A_j,A'_j\in \cA_j,\\|A_j\cap A'_j|=i}}\Pr[\cS_j, \cS'_j \subset \cT]\\
			&=\sum_{i=2k_j}^{k+k_j}\sum_{\substack{A_j,A'_j\in \cA_j,\\|A_j\cap A'_j|=i}}\frac{1}{q^{\dim_{\ff_q}(\cS_j+\cS'_j)}}.
		\end{align*}
		So it suffices to calculate $\dim_{\ff_q}(\cS_j+\cS'_j)$ in order to estimate $\mathbb{E}[X_j^2]$. 
		Note that 
		\begin{align*}
			\dim_{\ff_q}(\cS_j+\cS'_j)&= \dim_{\ff_q} \cS_j +\dim_{\ff_q}\cS'_j - \dim_{\ff_q}(\cS_j\cap\cS'_j)\\
			&= 2k_j - \dim_{\ff_q}(\cS_j\cap\cS'_j).
		\end{align*}
		Assume $|A_j\cap A'_j|=i\geq 2k_j$. 
		{
			We claim that the intersection of $\cS_j$ and $\cS_j'$ is trivial if $i\leq k$. Indeed, assuming $i \leq k$, by Proposition~\ref{prop:s} we may choose the $k$-subset $B_j$ (resp., $B_j'$) of $A_j$ (resp., $A_j'$) in such a way that it contains $A_j\cap A_j'$. Note that $\cS_j$ (resp., $\cS_j'$) can be viewed as generated by the columns of $G_L^c$ with indices in $B_1+n,\ldots,B_{j-1}+n(j-1)$ and $A_j+nj$ (resp., $A_j'+nj$). By the choice of $B_j,B_j'$, the columns of $G_L^c$ with indices in $(A_j\setminus B_j)+nj,(A_j'\setminus B_j')+nj$ are all distinct. 
			Since the subspaces $\cS_j,\cS_j'$ can be viewed as generated by $(j+1)k+2k_j-i\leq(j+1)k$ distinct columns of $G_j^c$, it follows that $\cS_j\cap \cS_j'=\{0\}$.}
			
		{Now consider $i > k$. In this case we cannot choose $B_j,B_j'$ as above. Instead, we take $B_j=B_j'\subset (A_j\cap A_j')$.} Let $V_j$ and $V'_j$ be the matrices formed by the columns of $G_L^c$ with indices in $\big(\bigcup_{{s}=1}^{j-1}(B_{{s}}+{s}n)\big)\cup(A_j+jn)$ and $\big(\bigcup_{{s}=1}^{j-1}(B_{{s}}+{s}n)\big)\cup(A'_j+jn)$, respectively.
		It follows that $|V_j\cup V_j'|=(j+1)k+2k_j-i \leq (j+1)k$. 
		Let $\cV_j=\Span_{\ff_q}V_j$ and $\cV'_j=\Span_{\ff_q}V'_j$. Then by the MDP property, we have
		\begin{align}
			\dim_{\ff_q}(\cV_j+\cV'_j)=|V_j\cup V'_j| = (j+1)k+2k_j-i.\label{eq:1}
		\end{align} At the same time, we have
		\begin{align}
			\dim_{\ff_q}(\cV_j+\cV'_j) &= \dim_{\ff_q} \cV_j +\dim_{\ff_q}\cV'_j - \dim_{\ff_q}(\cV_j\cap\cV'_j)\nonumber\\
			&= 2jk+2k_j - (\dim_{\ff_q}(\cS_j\cap\cS'_j)+jk).\label{eq:2}
		\end{align}
		It {follows} from \eqref{eq:1} and \eqref{eq:2} that 
		\begin{align*}
			\dim_{\ff_q}(\cS_j\cap\cS'_j) = {i-k}. 
		\end{align*} 
		Therefore, $\dim_{\ff_q}(\cS_j+\cS'_j)=2k_j -\max\{i-k,0\}$ and 
		\begin{align*}
			\mathbb{E}[X_j^2]
			&=\sum_{i=2k_j}^{k+k_j}\sum_{\substack{A_j,A'_j\in \cA_j,\\|A_j\cap A'_j|=i}}\frac{1}{q^{2k_j -\max\{i-k,0\}}}\\
			&=\sum_{i=2k_j}^{k+k_j}\!\!\!\binom{{n-2k_j}}{k-k_j}\binom{k-k_j}{i-2k_j}\binom{{n-k-k_j}}{k+k_j-i}\frac{1}{q^{2k_j -\max\{i-k,0\}}}.
		\end{align*} 
		By the second moment method, we have 
		\begin{align*}
			\frac{1}{\Pr[X_j>0]} &\leq \frac{\mathbb{E}[X_j^2]}{\mathbb{E}[X_j]^2}\\
			&{={\sum_{i=2k_j}^{k}  \binom{k-k_j}{i-2k_j}\binom{n-k-k_j}{k+k_j-i} }
				\frac{1}{\binom{n-2k_j}{k-k_j} }
				+
				\sum_{i=k+1}^{k+k_j}\binom{k-k_j}{i-2k_j}\binom{n-k-k_j}{k+k_j-i}\frac{q^{i-k}}{\binom{n-2k_j}{k-k_j}}}\\
			& \leq 1+\sum_{i=k+1}^{k+k_j}\binom{k-k_j}{i-2k_j}\binom{{n-k-k_j}}{k+k_j-i}\frac{q^{i-k}}{\binom{{n-2k_j}}{k-k_j}}\\
			& = 1+O_k(1)\sum_{i=k+1}^{k+k_j}\frac{q^{i-k}}{n^{i-2k_j}}.
		\end{align*}
		
		Toward a contradiction, suppose that $q=\Omega_k(n^{k/\lceil k/L\rceil-1})$ does not hold. Since $k_j\leq \lceil k/L\rceil$, we have $q=o_k(n^{k/k_j-1})$. It follows that $\Pr[X_j>0]=1-o(1)$. 
		{Note that Proposition~\ref{prop:bound} implies $\Pr[x=0]=\frac{1}{q^k}=o(1)$.} 
		By the union bound, the probability that $\cT$ is a hyperplane in $\ff_q^k$ and all the $X_j$'s are at least one simultaneously is 
		\begin{align*}
			\Pr[x\neq 0 \land (X_j>0, j=1,\ldots,L)]
			&\geq 1 - \Pr[x = 0] - \sum_{j=1}^{L}\Pr[X_j\leq 0]\\
			&=1 - o(1).
		\end{align*}
		{This implies that if $q=\Omega_k(n^{k/\lceil k/L\rceil-1})$ does not hold, then there exists a hyperplane $\cT\subsetneq \ff_q^{k}$ such that $\cT\supset \cS_j$ for all $j=1,\ldots,L$, which is a contradiction to the fact that $\bigoplus_{j=1}^L \cS_j = \ff_q^{k}$.} Hence, $q=\Omega_k(n^{k/\lceil k/L\rceil-1})$.
	\end{proof}
	
	\begin{remark}\label{re:l}
		Although Lemma~\ref{le:vanishing-rate} is stated with respect to the maximum profile length $L$, it is clear that the proof can be readily modified to show that if $\cC$ is {a convolutional code} over $\ff_q$ with the $l$th column distance attaining the bound \eqref{eq:sb-c} with equality, then $q=\Omega_k(n^{k/\lceil k/l\rceil-1})$.
	\end{remark}
	So far we have assumed $k$ is a constant and thus the rate $R=k/n$ is vanishing as $n$ grows. 
	{The duality result of MDP convolutional codes implies that a similar bound for the field size also holds for MDP convolutional codes with rate approaching one. This observation is formulated in the following lemma.
	\begin{lemma}\label{le:rate-one}
		Let $r$ be a constant independent of $n$ and let $k=n-r$. Let $G(D)$ be a $k\times n$ basic and minimal generator matrix of an $(n,k,\delta)$ MDP convolutional code $\cC$ over $\ff_q$. If the maximum profile length is a constant $L>1$, then the field size satisfies
		\begin{align*}
			q=\Omega_r(n^{r/\lceil r/L\rceil-1}).
		\end{align*}
	\end{lemma}
	\begin{proof}
		By Theorem~\ref{thm:mdp-dual}, the dual code $\cC^{\perp}$ is an $(n,r,\delta)$ MDP convolutional code over $\ff_q$ that has the same maximum profile length $L$. Then it follows from Lemma~\ref{le:vanishing-rate} that $q=\Omega_r(n^{r/\lceil r/L\rceil-1})$.
	\end{proof}
	
	Given an MDP convolutional code of arbitrary rate, we may puncture the code or its dual code to appropriate coordinates and obtain a code with rate approaching one. This enables us to establish a bound for any rate.
	}
	\begin{theorem}\label{thm:main}
		Let $G(D)$ be a $k\times n$ basic and minimal generator matrix of an $(n,k,\delta)$ MDP convolutional code $\cC$ over $\ff_q$. If the maximum profile length is a constant $L\geq 1$, then 
		\begin{align*}
			q=\Omega_L(n^{L-1}).
		\end{align*}
	\end{theorem}
	
	\begin{proof}
		If $L=1$, then $q=\Omega_L(n^{L-1})$ holds trivially. Assume $L\geq 2$. Let us consider the case $k/n\geq 1/2$ first.
		
		By Lemma~\ref{le:minimal-matrices}, the matrix $\bar{G}$ of the highest order coefficients for $G(D)$ satisfies $\rank \bar{G} =k$. Let $I$ be a $k$-subset of $\{1,\ldots,n\}$ such that the square submatrix $\bar{G}_I$ has rank $k$.
		
		Let $J\subset \{1,\ldots,n\}$ be such that $J\supset I$ and let $G_J(D)$ be the $k\times |J|$ polynomial matrix obtained by puncturing $G(D)$ to the columns with indices in $J$. It follows that $\rank \bar{G}_J=k$. By Lemma~\ref{le:minimal-matrices}, $G_J(D)$ is minimal and has rank $k$. 
		Therefore, $G_J(D)$ generates a $(|J|,k,\delta)$ convolutional code $\cC_J$ over $\ff_q$. 
		{Note that $G_J(D)$ may not be basic in general, and therefore, $\cC_J$ may not have a polynomial parity-check matrix. However, the characterization of column distances in terms of generator matrices (Item~1 and 2 of Theorem~\ref{thm:cd}) still applies because the assumption of a basic generator matrix in Theorem~\ref{thm:cd} is only required to ensure the existence of a polynomial parity-check matrix so as to establish the characterization of column distances in terms of parity-check matrices (cf. \cite[Theorem~2.4]{gluesing2006strongly} and its proof \cite[Appendix~A]{gluesing2006strongly}).
		Therefore, by Item~2 of Theorem~\ref{thm:cd}}, the $L$th column distance of $\cC_J$ {attains} the bound \eqref{eq:sb-c} with equality {since the maximum profile length of $\cC$ is $L$}.
		{By assumption, $L$ is a constant independent of $n$, and thus we may take} $|J|=k+L$.
		Then it follows from {Lemma~\ref{le:rate-one}} and Remark~\ref{re:l} that the field size of {$\cC_J$} satisfies $q=\Omega_L((k+L)^{L-1})=\Omega_L(n^{L-1})$, which also holds for $\cC$. 
		
		Now consider the case $k/n<1/2$. It is clear that the dual code $\cC^\perp$ is an $(n,n-k,\delta)$ convolutional code of rate at least $1/2$, whose $L$th column distance attains the bound \eqref{eq:sb-c}. Therefore, we may apply the above puncturing technique for the case $k/n\geq 1/2$ to the code $\cC^\perp$. It follows that $q=\Omega_L((n-k+L)^{L-1})=\Omega_L(n^{L-1})$.
	\end{proof}
	
	\begin{remark}
		By Definition~\ref{def:sMDS-MDP}, 
		when $n-k \mid \delta$, Theorem~\ref{thm:main} also serves as a lower bound on the field size of strongly-MDS convolutional codes.
	\end{remark}
	
	\section{The improved construction}
	In this section we present an explicit construction of $(n,k,\delta)$ MDP convolutional codes with $L=1$ over a finite field of size $O(n^{\delta})$. More precisely, we show that for any rate $R=k/n\in(0,1)$ such that $R\neq 1/2$, there exists an explicit MDP convolutional code with rate $R$ and degree $\delta=\min\{k,n-k\}$ over a finite field of size $O(n^{\delta})$. The idea behind our construction is similar to \cite{chen2022convolutional} that is based on the theory of skew polynomials, whereas we select a different collection of evaluation points for the skew polynomials that arise in the construction. In contrast to conventional polynomials, the product of skew polynomials is not commutative, and this gives rise to distinctive properties of skew polynomials. To facilitate the discussion of our code construction, we gather some basic results of skew polynomials in the following subsection. We refer the readers to \cite{martinez2022codes} for a more detailed discussion of skew polynomials and their applications.
	
	\subsection{Skew polynomials and skew Vandermonde matrices}
	
	Let $\ff_{q^t}$ be a finite field of size $q^t$ where $q$ is a prime power. Let $\sigma\colon \ff_{q^t}\to\ff_{q^t}$ be the Frobenius automorphism. Namely, $\sigma(a)=a^q$ for any $a\in\ff_{q^t}$.
	
	\begin{definition}[Ring of skew polynomials]\label{def:skew-polynomial}
		Let $\ff_{q^t}[x;\sigma]$ be the ring of skew polynomials with indeterminate in $x$ and coefficients in $\ff_{q^t}$, where addition in $\ff_{q^t}[x;\sigma]$ is coefficient wise and multiplication in $\ff_{q^t}[x;\sigma]$ is distributive and satisfies that for any $a\in\ff_{q^t}$
		\begin{align*}
				xa=\sigma(a)x.
		\end{align*} 
	\end{definition}
	
	We note that skew polynomials can be defined more generally. For example, $\sigma$ can be taken to any field automorphism of $\ff_{q^t}$ in general. For the purpose of presenting our code construction, we will only consider the class of skew polynomials given in Definition~\ref{def:skew-polynomial}. 
	
	As the product of skew polynomials is not commutative, evaluation of skew polynomials over elements in a finite field should be done properly. 
	
	\begin{lemma}[Evaluation of skew polynomials]\label{le:eval}
		Let $f(x)=\sum_{i=0}f_ix^i\in\ff_{q^t}[x;\sigma]$. Then the evaluation of $f(x)$ at any $a\in\ff_{q^t}$ is given by
		\begin{align*}
				f(a) = \sum_{i}f_iN_i(a),
			\end{align*}
		where $N_i(a):=a^\frac{q^i-1}{q-1}$ for $i\in\mathbb{N}$.
	\end{lemma}
	
	Different from conventional polynomials, a skew polynomial $f(x)\in\ff_{q^t}[x;\sigma]$ may have more than $\deg f$ roots in $\ff_{q^t}$. Moreover, the roots may belong to distinct \emph{conjugacy classes} in $\ff_{q^t}$ induced by the field automorphism $\sigma$.
	
	\begin{definition}[Conjugates]
		Let $a,b\in\ff_{q^t}$. Then $b$ is a conjugate of $a$ with respect to the field automorphism $\sigma$ if there exists $\beta\in\ff_{q^t}^*$ such that 
		\begin{align*}
				b=\sigma(\beta)a\beta^{-1}.
			\end{align*}
		We also call $b$ the $\beta$-conjugate of $a$ with respect to $\sigma$ and write $b={}^\beta a$ for brevity.
	\end{definition}
	
	The notion of conjugacy above defines an equivalence relation in $\ff_{q^t}$, and thus we can partition $\ff_{q^t}$ into distinct conjugacy classes. 
	For any $a\in\ff_{q^t}$, let us denote by $C_{a}^{\sigma}=\{ {}^\beta a \mid \beta\in\ff_{q^t}^*\}$ the conjugacy class with representative $a\in\ff_{q^t}$. The choice of $\sigma$ being the Frobenius automorphism induces the following conjugacy classes.
	
	\begin{proposition}[Conjugacy classes]\label{prop:conjugacy}
		Let $\gamma$ be a primitive element of $\ff_{q^t}$. Then $\{C_{0}^{\sigma},C_{\gamma^0}^{\sigma},C_{\gamma^1}^{\sigma},\ldots,C_{\gamma^{q-2}}^{\sigma}\}$ is a partition of $\ff_{q^t}$. Moreover, $|C_{0}^{\sigma}|=1$ and $|C_{\gamma^i}^{\sigma}|=\frac{q^t-1}{q-1}$ for $i=0,\ldots,q-2$.
	\end{proposition}
	
	The set of roots in $\ff_{q^t}$ of a skew polynomial $f(x)\in\ff_{q^t}[x;\sigma]$ has interesting structures closely related to the conjugacy classes of $\ff_{q^t}$.
	
	\begin{theorem}[Roots of skew polynomials]\label{thm:skew-dim}
		Let $f(x)\in\ff_{q^t}[x;\sigma]$ be a nonzero skew polynomial and $\Omega$ be the set of roots of $f(x)$ in $\ff_{q^t}$. Further, let $\bigcup_{i}\Omega_i=\Omega$ be the partition of $\Omega$ into conjugacy classes and let $\cS_{i}=\{\beta \mid {}^\beta a_i\in \Omega_i\}\cup\{0\}$ where $a_i$ is some fixed representative in $\Omega_i$. Then $\cS_{i}$ is a vector space over $F_i$ where $F_i=\ff_{q^t}$ if $\Omega_i=C_0^{\sigma}$ and otherwise $F_i=\ff_q$. Moreover, 
		\begin{align*}
				\sum_{i}\dim_{F_i} \cS_{i} \leq \deg f.
			\end{align*}
	\end{theorem}
	
	A skew Vandermonde matrix evaluates a skew polynomial $f(x)\in\ff_{q^t}[x;\sigma]$ at a set of points $\Omega\subset\ff_{q^t}$. The non-singularity of skew Vandermonde matrices is implied by Theorem~\ref{thm:skew-dim}.
	\begin{definition}[Skew Vandermonde matrices]\label{def:skew-vandermonde}
		Let $k$ be a positive integer and let $\Omega=\{a_1,\ldots,a_n\}\subset \ff_{q^t}$. The $k\times n$ skew Vandermonde matrix with respect to $\Omega$, denoted by $V_k(\Omega)$, is defined to be 
		\begin{align*}
				V_k(\Omega)=\begin{pmatrix}
						N_0(a_1) & N_0(a_2) & \cdots & N_0(a_n)\\
						N_1(a_1) & N_1(a_2) & \cdots & N_1(a_n)\\
						\vdots & \vdots & \ddots & \vdots\\
						N_{k-1}(a_1) & N_{k-1}(a_2) & \cdots & N_{k-1}(a_n)
					\end{pmatrix}.
			\end{align*} 
	\end{definition}
	
	The following corollary is a consequence of Theorem~\ref{thm:skew-dim}.
	
	\begin{corollary}[Rank of skew Vandermonde matrices]\label{coro:skew-vandermonde}
		Let $\Omega$ be an $n$-subset of $\ff_{q^t}$ and $\bigcup_{i}\Omega_i=\Omega$ be the partition of $\Omega$ into conjugacy classes. Let $l_i=|\Omega_i|$ and fix $a_i\in \Omega_i$. Suppose further that $\Omega_i=\{{}^{\beta_{ij}}a_i\mid j=1,\ldots,l_i\}$. Then $V_n(\Omega)$ has full rank if and only if 
		for each $i$ 
		the set $\{\beta_{ij}\mid j=1,\ldots,l_i\}$ is linearly independent over $F_i$ where $F_i=\ff_{q^t}$ if $\Omega_i=C_0^{\sigma}$ and otherwise $F_i=\ff_q$.
	\end{corollary}
	
	Note that evaluation of skew polynomials does not preserve linearity. Nevertheless, the following result provides a way to linearize the evaluation of skew polynomials on any conjugacy classes.
	
	\begin{lemma}[Linearization]\label{le:skew-linearized}
		Let $f(x)\in\ff_{q^t}[x;\sigma], a\in\ff_{q^t}$, and $\beta\in\ff_{q^t}^*$. Then
		$\cD_{f,a}(\beta):=f({}^\beta a)\beta$ is an $\ff_q$-linear map. In other words, for any $\lambda_1,\lambda_2\in\ff_q$ and $\beta_1,\beta_2\in\ff_{q^t}$, we have $\cD_{f,a}(\lambda_1\beta_1+\lambda_2\beta_2)=\lambda_1\cD_{f,a}(\beta_1)+\lambda_2\cD_{f,a}(\beta_2)$.
	\end{lemma}
	
	Proposition~\ref{prop:conjugacy}, Corollary~\ref{coro:skew-vandermonde}, and Lemma~\ref{le:skew-linearized} together imply the following.
	
	\begin{theorem}\label{thm:rank}
			Let $\Omega$ be an $n$-subset of $\ff_{q^t}^*$ and $\bigcup_{i}^{m}\Omega_i=\Omega$ be the partition of $\Omega$ into conjugacy classes where $m\leq q-2$. Let $l_i=|\Omega_i|$ and let $\gamma$ be a primitive element of $\ff_{q^t}$. Suppose $a_i=\gamma^{t_i}\in \Omega_i$ and that $\Omega_i=\{{}^{\beta_{i,j}}a_i\mid j=1,\ldots,l_i\}$. Then the matrix
			\begin{align*}
				\tilde{V}_n(\Omega):=\left(\begin{array}{ccc|c|ccc}
					N_0({}^{\beta_{1,1}}a_1)\beta_{1,1}     & \cdots & N_0({}^{\beta_{1,l_1}}a_1)\beta_{1,l_1}     & \cdots & N_0({}^{\beta_{m,1}}a_m)\beta_{m,1}     & \cdots & N_0({}^{\beta_{m,l_m}}a_m)\beta_{m,l_m}     \\
					N_1({}^{\beta_{1,1}}a_1)\beta_{1,1}     & \cdots & N_1({}^{\beta_{1,l_1}}a_1)\beta_{1,l_1}     & \cdots & N_1({}^{\beta_{m,1}}a_m)\beta_{m,1}     & \cdots & N_1({}^{\beta_{m,l_m}}a_m)\beta_{m,l_m}     \\
					\vdots                                  & \ddots & \vdots                                      & \cdots & \vdots                                  & \ddots & \vdots                                      \\
					N_{n-1}({}^{\beta_{1,1}}a_1)\beta_{1,1} & \cdots & N_{n-1}({}^{\beta_{1,l_1}}a_1)\beta_{1,l_1} & \cdots & N_{n-1}({}^{\beta_{m,1}}a_m)\beta_{m,1} & \cdots & N_{n-1}({}^{\beta_{m,l_m}}a_m)\beta_{m,l_m}
				\end{array}\right).
			\end{align*} 
			has full rank if and only if for each $i=1,\ldots,m$
			the set $\{\beta_{i,j}\mid j=1,\ldots,l_i\}$ is linearly independent over $\ff_q$.
	\end{theorem}
	
	\subsection{Construction of MDP convolutional codes with $L=1$}
In this subsection we present a family of $(n,k,\delta)$ MDP convolutional codes with the maximum profile length $L=1$ over a finite field of size $O(n^{\delta})$ where $\delta=\min\{k,n-k\}$. Our construction can viewed as an improved version of code construction in \cite{chen2022convolutional}. Similar to \cite{chen2022convolutional}, the generator matrix of our construction is based on skew Vandermonde matrices. The improvement of the field size comes from a clever trick in choosing the defining elements of the underlying skew Vandermonde matrices.

\begin{construction}\label{con:1}
Let $q\geq\max\{3,n\}$ be a prime power and let 
$\gamma$ be a primitive element of the finite field $\ff_{q^k}$.
Further, let $\{\lambda_1,\ldots,\lambda_n\}\subset\ff_q$ be $n$ distinct elements.
For $j=0,1$ denote $\gamma_j=\gamma^j$ and 
define $\alpha_{j,i}\in\ff_{q^k}$ using a fixed basis for $\ff_{q^k}$ over $\ff_q$ as 
\begin{align}
	\alpha_{j,i}&=
	(
	\lambda_i^{(1-j)k},
	\lambda_i^{(1-j)k+1},
	\ldots,
	\lambda_i^{(2-j)k-1}
	)
	,\quad i=1,\ldots,n.\label{eq:alpha}
\end{align}
Let $\cC$ be an $(n,k)$ convolutional code over $\ff_{q^k}$ with generator matrix $G(D)=G_0+G_1D$ where $G_j,j=0,1$ are given by
\begin{align}
	G_j=\begin{pmatrix}
			N_0({}^{\alpha_{j,1}} \gamma_j)\alpha_{j,1} & N_0({}^{\alpha_{j,2}} \gamma_j)\alpha_{j,2} & \cdots & N_0({}^{\alpha_{j,n}} \gamma_j)\alpha_{j,n}\\
			N_1({}^{\alpha_{j,1}} \gamma_j)\alpha_{j,1} & N_1({}^{\alpha_{j,2}} \gamma_j)\alpha_{j,2} & \cdots & N_1({}^{\alpha_{j,n}} \gamma_j)\alpha_{j,n}\\
			\vdots & \vdots & \ddots & \vdots\\
			N_{k-1}({}^{\alpha_{j,1}} \gamma_j)\alpha_{j,1} & N_{k-1}({}^{\alpha_{j,2}} \gamma_j)\alpha_{j,2} & \cdots & N_{k-1}({}^{\alpha_{j,n}} \gamma_j)\alpha_{j,n}
		\end{pmatrix}.\nonumber
\end{align}
\end{construction}

Note that the elements $\alpha_{j,i}$ are defined in such a way \eqref{eq:alpha} that one may regard $\beta_i:=(\alpha_{1,i},\alpha_{0,i})$ as a length-$(2k)$ vector $(1,\lambda_i,\lambda_i^2,\ldots,\lambda_i^{2k-1})$ over $\ff_q$. From this perspective, since $\lambda_1,\ldots,\lambda_n$ are $n$ distinct elements in $\ff_q$, any $2k$ vectors of the set $\{\beta_1,\ldots,\beta_n\}$ are linearly independent over $\ff_q$. This ``concatenated independence'' property of the elements $\beta_i$ will be instrumental in establishing the MDP property of Construction~\ref{con:1}. In fact, this trick of choosing the field elements to satisfy the ``concatenated independence'' property is used in {\cite{cai2021construction,gopi2022improved}}, where the task is to show certain block matrices are of full rank, and also in the recent construction of MDP convolutional codes \cite{luo2023construction}.

Before proving the MDP property of Construction~\ref{con:1}, let us first observe a few simple properties for the generator matrix $G(D)$ that will be useful later. 

\begin{proposition}\label{prop:skew-vandermonde-full-rank}
	Any $k\times k$ submatrix of $G_j$ defined in Construction~\ref{con:1} has full rank.
\end{proposition}
\begin{proof}
	By \eqref{eq:alpha}, any $k$ elements of $\{\alpha_{j,i}\mid i=1,\ldots,n\}$ are linearly independent over $\ff_q$. 
	Since ${}^{\alpha_{j,1}}\gamma_j,\ldots,{}^{\alpha_{j,n}}\gamma_j$ are in the same conjugacy class $C_{\gamma_j}^{\sigma}$, it follows from Theorem~\ref{thm:rank} that any $k\times k$ submatrix of $G_j$ has full rank.
\end{proof}

\begin{proposition}\label{prop:minimal}
	The $k\times n$ generator matrix $G(D)$ defined in Construction~\ref{con:1} is minimal.
\end{proposition}
\begin{proof}
	It is clear that the matrix of the highest order coefficients for $G(D)$ is $\bar{G}=G_1$. Moreover, by Proposition~\ref{prop:skew-vandermonde-full-rank}, $\bar{G}=G_1$ has rank $k$. Therefore, by Lemma~\ref{le:minimal-matrices}, $G(D)$ is minimal.
\end{proof}

In the following we first prove that $\cC$ is an MDP convolutional code whenever the rate $k/n< 1/2$. Once this is done, the result is extended for construction of MDP convolutional codes with rate $k/n>1/2$ using Theorem~\ref{thm:mdp-dual}. {We note that the starting point of our proof here is similar to the proof of Lemma~\ref{le:vanishing-rate}. Namely, we start with an arbitrary full-size square submatrix of $G_L^c$ that has nontrivial determinant and reduce the submatrix to a proper upper triangular form with square blocks on its diagonal, facilitating the proof of the MDP property.} 

\begin{theorem}\label{thm:m1}
	Assume $n>2k$. The $(n,k)$ code $\cC$ defined in Construction~\ref{con:1} is an $(n,k,\delta=k)$ MDP convolutional code over $\ff_{q^k}$ with $L=1$. 
\end{theorem}

\begin{proof}
By Proposition~\ref{prop:minimal}, $G(D)$ is minimal, and thus the degree of $\cC$ is $\delta=k$. Since $n>2k$, we have $L=\lfloor\frac{\delta}{k}\rfloor+\lfloor\frac{\delta}{n-k}\rfloor=1$. Note that $G(D)$ has row degree $\nu_i=1=\delta/k$ for all $i$. Therefore, by Lemma~\ref{le:mdp-generic}, for the code $\cC$ to be MDP, it suffices to show that $G_1^c$ has the MDP property. 

Let $A_{i}\subset\{1,\ldots,n\},i=0,1$ be such that $|A_0|\leq k$ and $|A_0|+|A_1|=2k$. Let $P$ be the $2k\times 2k$ matrix formed by columns of $G_1^c$ with indices in the set $A_{0}\cup(A_{1}+n)$ where $A_{1}+n$ means the set obtained by adding every element of $A_{1}$ by the integer $n$. 
For clarity, let us write the matrix $P$ explicitly as
\begin{align*}
	P=\begin{pmatrix}
		G_{0,A_0} & G_{1,A_1} \\
		0         & G_{0,A_1}
	\end{pmatrix},
\end{align*} where $G_{i,A_j}$ is the submatrix of $G_i$ with column indices in $A_j$.

We would like to show that the matrix $P$ has full rank. Note that if $|A_0|=|A_1|=k$, then $\rank P=2k$ since by Proposition~\ref{prop:skew-vandermonde-full-rank}, $G_{0,A_0},G_{1,A_1},G_{0,A_1}$ are $k\times k$ matrices of full rank.
Assume $|A_0|<k$ and $|A_1|>k$. Let $B_1$ be a $k$-subset of $A_1$. By Proposition~\ref{prop:skew-vandermonde-full-rank}, we have $\rank G_{0,B_1}=k$. Furthermore, since $\{\alpha_{0,i}\mid i\in B_1\}$ is linearly independent over $\ff_q$, by Lemma~\ref{le:skew-linearized}, there exists a $k\times (|A_1|-k)$ matrix $T$ over $\ff_q$ such that $G_{0,A_1\setminus B_1}=G_{0,B_1}T$.
It follows that $P$ can be transformed by elementary column operations to 
\begin{align*}
	\begin{pmatrix}
		G_{0,A_0} & {S} & G_{1,B_1} \\
		0         & 0   & G_{0,B_1}
	\end{pmatrix},
\end{align*} where ${S}=G_{1,A_1\setminus B_1}-G_{1,B_1}T$ is an $k\times (|A_1|-k)$ matrix over $\ff_{q^k}$. So the problem of establishing the non-singularity of $P$ boils down to showing the matrix $\begin{pmatrix}
G_{0,A_0} & {S}
\end{pmatrix}$ has full rank.

By Lemma~\ref{le:skew-linearized}, since $T$ is a matrix over $\ff_q$, the matrix $S$ is also a skew Vandermonde matrix where each column is scaled by some nonzero element in $\ff_{q^k}$. More precisely, setting $k_1 = |A_1|-k$, we have
\begin{align}
	S=V_k(\{{}^{\mu_1}\gamma,\ldots,{}^{\mu_{k_1}}\gamma\})\begin{pmatrix}
		\mu_1 & & \\
		& \ddots & \\
		& & \mu_{k_1}
	\end{pmatrix},\label{eq:col-op}
\end{align} for some $\mu_1,\ldots,\mu_{k_1}$ that will be made explicit shortly. Observe that if $\mu_1,\ldots,\mu_{k_1}$ are linearly independent over $\ff_q$, then $\rank S=k_1$. Moreover, by Proposition~\ref{prop:skew-vandermonde-full-rank}, we also have $\rank G_{0,A_0}=|A_0|$. As a consequence, it would follow from Theorem~\ref{thm:rank} that the matrix $\begin{pmatrix}
G_{0,A_0} & {S}
\end{pmatrix}$ has full rank.

Next, let us show that $\mu_1,\ldots,\mu_{k_1}$ are indeed linearly independent over $\ff_q$, and this is where the trick of defining $\alpha_{j,i}$ as in \eqref{eq:alpha} comes into play.
Without loss of generality, let us assume $B_1=\{1+k_1,\ldots,k+k_1\}$ and $A_1=\{1,\ldots,k+k_1\}$. 
{
Note that one may view $\beta_i=(\alpha_{1,i},\alpha_{0,i})$ as a length-$(2k)$ vector $(1,\lambda_i,\lambda_i^2,\ldots,\lambda_i^{2k-1})$ over $\ff_q$.
Define the following} $2k\times (k+k_1)$ matrix $V$ over $\ff_q$:
\begin{align*}
	V:=
	\begin{pmatrix}
		1                & 1                & \ldots & 1                      \\
		\lambda_1        & \lambda_2        & \ldots & \lambda_{k+k_1}        \\
		\vdots           & \vdots           & \ddots & \vdots                 \\
		\lambda_1^{2k-1} & \lambda_2^{2k-1} & \ldots & \lambda_{k+k_1}^{2k-1}
	\end{pmatrix}.
\end{align*}
Clearly, we have $\rank V= k+k_1$ as $\lambda_1,\ldots,\lambda_{k+k_1}$ are all distinct. Applying the same elementary column operations that lead to \eqref{eq:col-op}, we have
\begin{align}
	\left(
	\begin{array}{cccc}
	U & \begin{matrix}
		1                      & \ldots & 1                      \\
		\vdots                 & \ddots & \vdots                 \\
		\lambda_{1+k_1}^{k-1} & \ldots & \lambda_{k+k_1}^{k-1}
	\end{matrix}\\
	\midrule
	0 & \begin{matrix}
		\lambda_{1+k_1}^k        & \ldots & \lambda_{k+k_1}^k        \\
		\vdots                 & \ddots & \vdots                 \\
		\lambda_{1+k_1}^{2k-1} & \ldots & \lambda_{k+k_1}^{2k-1}
	\end{matrix}
	\end{array}
	\right) = V \begin{pmatrix}
		I_{k_1} & 0 \\
		-T & I_{k}
	\end{pmatrix},\label{eq:col-op-2}
\end{align} where $U$ is a $k\times k_1$ matrix over $\ff_q$ and $I_s$ is the $s\times s$ identity matrix. Since $V$ has full rank and so does the left-hand side of \eqref{eq:col-op-2}, it follows that $\rank U=k_1$. Let us write $U=\begin{pmatrix}
u_1 & \ldots & u_{k_1}
\end{pmatrix}$ where $u_i$ is the $i$th column of $U$ for $i=1,\ldots,k_1$. By Lemma~\ref{le:skew-linearized}, $u_1,\ldots,u_{k_1}$ are the representation of $\mu_1,\ldots,\mu_{k_1}$ as length-$k$ vectors over $\ff_q$. Thus, $\mu_1,\ldots,\mu_{k_1}$ are linearly independent over $\ff_q$ and $\rank S=k_1$. Furthermore, since $\gamma_j,j=0,1$ belong to distinct conjugacy classes, it follows from Theorem~\ref{thm:rank} that the matrix $\begin{pmatrix}
G_{0,A_0} & {S}
\end{pmatrix}$ has full rank. Hence, $G_1^c$ has the MDP property and the code $\cC$ is an $(n,k,\delta=k)$ MDP convolutional code over $\ff_{q^k}$ with maximum profile length $L=1$.
\end{proof}

The above result gives an explicit construction of MDP convolutional code with rate $k/n<1/2$ and degree $\delta=k$. Applying Theorem~\ref{thm:mdp-dual}, one can obtain MDP convolutional codes with higher rate and the same degree:

\begin{corollary}\label{coro:m1-dual}
	Assume $n>2k$. Let $\cC^\perp$ be the dual code of the $(n,k)$ code $\cC$ defined in Construction~\ref{con:1}.
	Then $\cC^\perp$ is an $(n,n-k,\delta=k)$ MDP convolutional code over $\ff_{q^k}$.
\end{corollary}

For a given rate $R=k/n$, it is often desirable to construct convolutional codes with smaller degree, which results in better decoding complexity. Combining Theorem~\ref{thm:m1} and Corollary~\ref{coro:m1-dual}, we have the following result.

\begin{corollary}\label{coro:m1-k}
	Assume $n\neq 2k$. There exists a family of $(n,k,\delta=\min\{k,n-k\})$ MDP convolutional codes that can be constructed explicitly over a finite field of size $\Theta(n^{\delta})$. 
\end{corollary}

\begin{remark}
	Note that in the case of $n>2k$, by Corollary~\ref{coro:m1-k}, we can construct a family of $(n,k,n-k)$ MDP convolutional codes with $L=1$. At the same time, since $\delta=n-k$, by Definition~\ref{def:sMDS-MDP}, we have $M=L=1$, and thus the $(n,k,n-k)$ MDP convolutional codes are also strongly-MDS convolutional codes.
\end{remark}

\section{Open problems}
In this paper we have presented the first nontrivial lower bound for $(n,k)$ MDP convolutional codes with maximum profile length $L\geq 1$. The bound suggests that a finite field of size at least polynomial in $n$ is inevitable for constructing $(n,k)$ MDP convolutional codes with maximum profile length $L\geq 3$. 
Additionally, we have also presented an explicit construction of $(n,k)$ MDP convolutional code with maximum profile length $L=1$ over a finite field of size $O(n^{\min\{k,n-k\}})$, improving upon the previous constructions with the same profile length.
However, there is an enormous gap between the lower bound and the upper bounds implied by known constructions (see Table~\ref{tab:comp}). Closing this gap is an interesting direction for future research. 

We note that our proof for the lower bound only examines a subset of all full-size minors that are required to be nonzero by the MDP property. It is natural to expect that one could further improve the lower bound by fully utilizing the MDP property.
Improving the upper bounds is also of great interest. Thus far, constructions of MDP convolutional codes with general parameters all require a finite field of size at least exponential in $n$. Are there explicit constructions of MDP convolutional codes with general parameters over a field of size polynomial in $n$, or even linear in $n$ for the case $L=1,2$? 

\section*{Acknowledgements}
The author is thankful to the anonymous reviewers for their careful reading of the manuscript and their insightful comments and suggestions.

\appendices
\addtocontents{toc}{\protect\setcounter{tocdepth}{0}}
{
\section{Proof of proposition~\ref{prop:s}}\label{app:a}
\begin{proof}
	We prove by induction on $i$. For $i=1,\ldots,j$, denote by $U_i$ the $ki\times(ki+\sum_{s=1}^{i}k_s)$ matrix
	\begin{align*}
		\begin{pmatrix}
			G_{0,A_{1}} & G_{1,A_{2}}   & \dots  & G_{i-1,A_i} \\
			& G_{0,A_{2}}   & \dots  & G_{i-2,A_i} \\
			&                 & \ddots & \vdots      \\
			&                 &        & G_{0,A_i}
		\end{pmatrix}.
	\end{align*}
	
	Let us first establish the induction basis. Note that $U_1=G_{0,A_1}$. Since $\rank G_{0,B_1}=k$, there exists a $(k+k_1)\times k_1$ matrix $T_1$ of full rank such that 
	\begin{align*}
		\begin{pmatrix}
			G_{1,A_1}\\
			U_1
		\end{pmatrix} T_1
		=
		\begin{pmatrix}
			S_1\\
			0
		\end{pmatrix}.
	\end{align*}
	Similarly, there exists a $(k+k_1)\times k_1$ matrix $T_1'$ of full rank such that 
	\begin{align*}
		\begin{pmatrix}
			G_{1,A_1}\\
			U_1
		\end{pmatrix} T_1'
		=
		\begin{pmatrix}
			S_1'\\
			0
		\end{pmatrix}.
	\end{align*}
	Note that $\rank U_1=k$ and $\dim_{\ff_q} \ker U_1= k+k_1-\rank U_1=k_1$. Thus, $\Span_{\ff_q} T_1 =\Span_{\ff_q} T_1' = \ker U_1$. Therefore, $T_1$ and $T_1'$ form two bases for $\ker U_1$, and there exists a $k_1\times k_1$ invertible matrix $V_1$ such that $T_1=T_1'V_1$. It follows that $S_1=S_1'V_1$. Hence, $\Span_{\ff_q} S_1=\Span S_1'$.
	
	Now let $1 < i \leq j$ and assume $\Span_{\ff_q} S_a=\Span S_a'$ for $a=1,\ldots,i-1$.
	Since $\rank G_{0,B_1}=\ldots=\rank G_{0,B_i}=k$, there exists a $(ki+\sum_{s=1}^{i}k_s)\times \sum_{s=1}^{i}k_s$ matrix $T_i$ of full rank such that 
	\begin{align*}
		\begin{pmatrix}
			\begin{matrix}
				G_{1,A_{1}} & G_{2,A_{2}}   & \dots  & G_{i,A_i}
			\end{matrix}\\
			\begin{matrix}
				U_i
			\end{matrix}
		\end{pmatrix} T_i
		=\begin{pmatrix}
			\begin{matrix}
				S_1 & S_2 & \dots  & S_i    
			\end{matrix}\\
			\begin{matrix}
				0
			\end{matrix}
		\end{pmatrix}.
	\end{align*}
	Similarly, there exists a $(ki+\sum_{s=1}^{i}k_s)\times \sum_{s=1}^{i}k_s$ matrix $T_i'$ of full rank such that 
	\begin{align*}
		\begin{pmatrix}
			\begin{matrix}
				G_{1,A_{1}} & G_{2,A_{2}}   & \dots  & G_{i,A_i}
			\end{matrix}\\
			\begin{matrix}
				U_i
			\end{matrix}
		\end{pmatrix} T_i'
		=\begin{pmatrix}
			\begin{matrix}
				S_1' & S_2' & \dots  & S_i'     
			\end{matrix}\\
			\begin{matrix}
				0
			\end{matrix}
		\end{pmatrix}.
	\end{align*}
	Since $\rank U_i=ki$, we have $\dim_{\ff_q} \ker U_i= ki+\sum_{s=1}^{i}k_s-\rank U_i=\sum_{s=1}^{i}k_s$. Thus, $\Span_{\ff_q} T_i =\Span_{\ff_q} T_i' = \ker U_i$. It follows that there exists an invertible matrix $V_i$ such that $ 
	\begin{pmatrix}
		S_1 & S_2 & \dots & S_i
	\end{pmatrix}
	=
	\begin{pmatrix}
		S_1' & S_2' & \dots & S_i'
	\end{pmatrix} V_i
	$.
	Furthermore, we have $\Span_{\ff_q} \begin{pmatrix}
		S_1 & S_2 & \dots & S_i
	\end{pmatrix}=\Span_{\ff_q} \begin{pmatrix}
		S_1' & S_2' & \dots & S_i'
	\end{pmatrix}$. Since $\bigoplus_{a=1}^i\Span_{\ff_q}S_a=\bigoplus_{a=1}^i\Span_{\ff_q}S_a'$, by the induction hypothesis, we obtain $\Span_{\ff_q} S_i=\Span_{\ff_q} S_i'$. Hence, we conclude by induction that $\Span_{\ff_q} S_i=\Span_{\ff_q} S_i'$ for $i=1,\ldots,j$. 
\end{proof}
}

\bibliographystyle{IEEEtran}
\bibliography{MDP}

\begin{thebibliography}{10}
\providecommand{\url}[1]{#1}
\csname url@samestyle\endcsname
\providecommand{\newblock}{\relax}
\providecommand{\bibinfo}[2]{#2}
\providecommand{\BIBentrySTDinterwordspacing}{\spaceskip=0pt\relax}
\providecommand{\BIBentryALTinterwordstretchfactor}{4}
\providecommand{\BIBentryALTinterwordspacing}{\spaceskip=\fontdimen2\font plus
\BIBentryALTinterwordstretchfactor\fontdimen3\font minus
  \fontdimen4\font\relax}
\providecommand{\BIBforeignlanguage}[2]{{%
\expandafter\ifx\csname l@#1\endcsname\relax
\typeout{** WARNING: IEEEtran.bst: No hyphenation pattern has been}%
\typeout{** loaded for the language `#1'. Using the pattern for}%
\typeout{** the default language instead.}%
\else
\language=\csname l@#1\endcsname
\fi
#2}}
\providecommand{\BIBdecl}{\relax}
\BIBdecl

\bibitem{johannesson2015fundamentals}
R.~Johannesson and K.~S. Zigangirov, \emph{Fundamentals of convolutional
  coding}.\hskip 1em plus 0.5em minus 0.4em\relax John Wiley \& Sons, 2015.

\bibitem{mceliece1993general}
\BIBentryALTinterwordspacing
R.~J. McEliece and R.~P. Stanley, ``The general theory of convolutional
  codes,'' \emph{The Telecommunications and Data Acquisition Report}, 1993.
  [Online]. Available:
  \url{https://ntrs.nasa.gov/api/citations/19930020406/downloads/19930020406.pdf}
\BIBentrySTDinterwordspacing

\bibitem{rosenthal1999maximum}
J.~Rosenthal and R.~Smarandache, ``Maximum distance separable convolutional
  codes,'' \emph{Applicable Algebra in Engineering, Communication and
  Computing}, vol.~10, no.~1, pp. 15--32, 1999.

\bibitem{gluesing2006strongly}
H.~Gluesing-Luerssen, J.~Rosenthal, and R.~Smarandache, ``Strongly-{MDS}
  convolutional codes,'' \emph{IEEE Transactions on Information Theory},
  vol.~52, no.~2, pp. 584--598, 2006.

\bibitem{hutchinson2005convolutional}
R.~Hutchinson, J.~Rosenthal, and R.~Smarandache, ``Convolutional codes with
  maximum distance profile,'' \emph{Systems \& Control Letters}, vol.~54,
  no.~1, pp. 53--63, 2005.

\bibitem{alfarano2020left}
G.~N. Alfarano and J.~Lieb, ``On the left primeness of some polynomial matrices
  with applications to convolutional codes,'' \emph{Journal of Algebra and Its
  Applications}, p. 2150207, 2020.

\bibitem{almeida2013new}
P.~Almeida, D.~Napp, and R.~Pinto, ``A new class of superregular matrices and
  {MDP} convolutional codes,'' \emph{Linear Algebra and its Applications}, vol.
  439, no.~7, pp. 2145--2157, 2013.

\bibitem{alfarano2020weighted}
G.~N. Alfarano, D.~Napp, A.~Neri, and V.~Requena, ``Weighted {R}eed-{S}olomon
  convolutional codes,'' \emph{Linear and Multilinear Algebra}, pp. 1--34,
  2023.

\bibitem{chen2022convolutional}
Z.~Chen, ``Convolutional codes with a maximum distance profile based on skew
  polynomials,'' \emph{IEEE Transactions on Information Theory}, vol.~68,
  no.~8, pp. 5178--5184, 2022.

\bibitem{luo2023construction}
G.~Luo, X.~Cao, M.~F. Ezerman, and S.~Ling, ``A construction of maximum
  distance profile convolutional codes with small alphabet sizes,'' \emph{IEEE
  Transactions on Information Theory}, vol.~69, no.~5, pp. 2983--2990, 2023.

\bibitem{lam1985general}
T.-Y. Lam, \emph{A general theory of {V}andermonde matrices}.\hskip 1em plus
  0.5em minus 0.4em\relax Center for Pure and Applied Mathematics, University
  of California, Berkeley, 1985.

\bibitem{lam1988vandermonde}
T.-Y. Lam and A.~Leroy, ``Vandermonde and {W}ronskian matrices over division
  rings,'' \emph{Journal of Algebra}, vol. 119, no.~2, pp. 308--336, 1988.

\bibitem{gopi2020maximally}
S.~Gopi, V.~Guruswami, and S.~Yekhanin, ``Maximally recoverable {LRC}s: A field
  size lower bound and constructions for few heavy parities,'' \emph{IEEE
  Transactions on Information Theory}, vol.~66, no.~10, pp. 6066--6083, 2020.

\bibitem{martinez2022codes}
U.~Mart{\'\i}nez-Pe{\~n}as, M.~Shehadeh, and F.~R. Kschischang, ``Codes in the
  sum-rank metric: Fundamentals and applications,'' \emph{Foundations and
  Trends{\textregistered} in Communications and Information Theory}, vol.~19,
  no.~5, pp. 814--1031, 2022.

\bibitem{cai2021construction}
H.~Cai, Y.~Miao, M.~Schwartz, and X.~Tang, ``A construction of maximally
  recoverable codes with order-optimal field size,'' \emph{IEEE Transactions on
  Information Theory}, vol.~68, no.~1, pp. 204--212, 2021.

\bibitem{gopi2022improved}
S.~Gopi and V.~Guruswami, ``Improved maximally recoverable {LRC}s using skew
  polynomials,'' \emph{IEEE Transactions on Information Theory}, vol.~68,
  no.~11, pp. 7198--7214, 2022.

\end{thebibliography}

\end{document}